\documentclass[11pt]{article}
\usepackage{amsmath}
\usepackage{amsfonts}
\usepackage{latexsym}
\usepackage{amssymb}
\usepackage{graphicx}
\usepackage{amsthm}
\usepackage{color}
\usepackage{chemarrow}
\usepackage{fullpage}
\usepackage{framed}
\usepackage{verbatim}
\usepackage{epsfig}
\usepackage{hyperref}
\usepackage{authblk}

\title{The Hot Bit I: The Szilard-Landauer Correspondence.}
\author{\small Manoj Gopalkrishnan\thanks{manojg@tifr.res.in}}
\affil{School of Technology and Computer Science,\\ Tata Institute of Fundamental Research, Mumbai, India.}
\date{15 November, 2013}
\begin{document}
\maketitle

\theoremstyle{definition}
\newtheorem{theorem}{\bf{Theorem}}[section]
\newtheorem{lemma}[theorem]{\bf{Lemma}}
\newtheorem{corollary}[theorem]{\bf{Corollary}}
\newtheorem{proofl}[theorem]{\bf{Proof}}
\newtheorem{open}{\bf{Open}}
\newtheorem{remark}[theorem]{Remark}
\newtheorem{definition}{\bf{Definition}}[section]
\newtheorem{conjecture}{\bf{Conjecture}}[section]
\newtheorem{postulate}{\bf{Postulate}}[section]

\theoremstyle{remark}
\newtheorem{example}{\bf{Example}}[section]
\newtheorem{thought}{\bf{Thought Experiment}}[section]
\newtheorem{notation}[definition]{\bf{Notation}}
\newtheorem{fact}{\bf{Fact}}
\newtheorem{note}[definition]{\bf{Note}}
\newcommand{\CE}{{\cal E}}
\newcommand{\MC}{\mathbb{C}}
\newcommand{\MR}{\mathbb{R}}
\newcommand{\MZ}{\mathbb{Z}}
\newcommand{\MN}{\mathbb{N}}
\newcommand{\MQ}{\mathbb{Q}}
\newcommand{\iton}{i=1,2,\ldots,n}
\newcommand{\jtom}{j=1,2,\ldots, m}
\newcommand{\MP}{\mathbb{R}_{>0}}
\newcommand{\bs}{\boldsymbol}
\newcommand{\op}{\operatorname}
\renewcommand\comment[1]{{\color{magenta} $\star$#1$\star$}}
\begin{center}

\end{center}

\begin{abstract}
We present a precise formulation of a correspondence between information and thermodynamics that was first observed by Szilard, and later studied by Landauer. The correspondence identifies available free energy with relative entropy, and provides a dictionary between information and thermodynamics. We precisely state and prove this correspondence. The paper should be broadly accessible since we assume no prior knowledge of information theory, developing it axiomatically, and we assume almost no thermodynamic background.
\end{abstract}

\section{Introduction}
This is the first in a series of articles on the thermodynamics of the bit. In future papers, we will pursue the question of lower  bounds to the cost of switching a bit. Here we aim to give as direct as possible a description of the connection between information theory and thermodynamics, by combining different pieces of arguments which have been described in information theory and statistical mechanics literature.

In 1929, Leo Szilard~\cite{Szilard1929} pointed out that the Maxwell's demon paradox could be resolved by a postulate that certain information processing tasks required energy. The issue of which information processing tasks to charge for, and how much, was clarified by Rolf Landuauer in 1961~\cite{landauer61irreversibility}, and Charles Bennett in 1987~\cite{bennett1987demons}. All these works showed a relation between information and thermodynamics in the special case of the Szilard engine: a cylindrical vessel containing one molecule of an ideal gas, immersed in a heat bath. 

It appears to be less well-known that the relation between information and thermodynamics can be precisely stated, and proved, in much more generality, as is shown in a very nice paper \cite{esposito2011second} by Esposito and Van den Broeck. To do this requires making precise what one means by information, and giving an identification between information and thermodynamic quantities. In this paper, we present a pedagogic introduction to these ideas. We develop in a more detailed manner the argument for the identification of relative entropy with information, and the argument for the identification of relative entropy with free energy.

Below we list our contributions. Many of these ideas are well-known within specialist domains, but to our knowledge have not appeared together before.
\begin{itemize}
\item In Section~\ref{sec:infotheo}, we introduce an operational interpretation of relative entropy in the context of ensembles to argue for a certain relative entropy as the right measure for the amount of information an observer knows about a system.
\item The meanings of certain information processing operations like ``Copy,'' have been given only implicitly in previous works. In Subsection~\ref{ss:bitops}, we show the ambiguity that can result from this, and take the opportunity to disambiguate by giving precise definitions in the language of Bernoulli random variables.
\item We precisely state and prove a general correspondence between free energy and relative entropy (Theorem~\ref{lem:SL}). 
\item The interpretation of relative entropy as information content together with the identification of relative entropy with free energy immediately implies that increasing the information one knows about a system requires a corresponding increase of free energy. Thus, the bounds on information processing tasks that were asserted in the works of Szilard and Landauer appear as consequences of the Second Law (Remark~\ref{rmk:landprin}).
\item The Szilard-Landauer correspondence allows an alternate formulation of the Second Law in terms of information. This gives a physical interpretation to ``Second Law'' theorems in information theoretic settings that assert the monotonic decrease of relative entropy (Remark~\ref{rmk:2ndlaw}), in a setting more general than that in \cite{esposito2011second}.
\end{itemize}

\section{Information Theory}\label{sec:infotheo}
The identification of information content in a system with relative entropy is a familiar idea in information theory. In this section, we briefly explain these ideas for the benefit of those readers who may not have been exposed to information theoretic ideas before.

\begin{notation}
Let $n\in \MZ_{\geq 1}$ be a positive integer.
\begin{enumerate}
\item When $S$ is a finite set, $\Delta^S = \{ p\in\MR^S_{\geq 0} \mid \sum_{x\in S} p_x = 1\}$ denotes the probability simplex in $\MR^S$.
\item $0 \log 0$ is understood as $\displaystyle\lim_ {x\to 0^+} x \log x = 0$.
\item All logarithms are to the natural base of Euler's constant $e$.
\item $k_B$ denotes Boltzmann's constant.
\end{enumerate}
\end{notation}

\begin{remark}
We proceed by a well-known axiomatic development of information theory. Let us suppose that there is a function called ``self-information'' that assigns positive real numbers to events. If two events are equally likely, then their occurrences should be equally surprising, so we postulate that \textbf{equally likely events are equally informative.} Hence, we can write the ``self-information'' or ``surprise'' function $\op{SI}:[0,1]\to \mathbb{R}_{\geq 0}$, so that it takes a probability and returns a positive real number. We postulate that more probable events are less surprising, so $\op{SI}$ is monotonically decreasing. We further postulate that \textbf{if two events of probabilities $p$ and $q$ respectively are independent then the information gained by learning of the joint occurrence of both events equals the information gained from the occurrence of one event plus the information gained from the occurrence of the other event.} This leads to the equation $SI(pq) = SI(p) + SI(q)$. With appropriate boundary conditions and up to choice of units, this forces $SI(p)=\log\frac{1}{p}$ for all $p\in[0,1]$~\cite[p.~84, 4.3.1]{dieudonné1969foundations}.
\end{remark}

In this paper, we will work with natural logarithms, so that one bit of information corresponds to a value of $\log 2$ nats.

\begin{definition}[Entropy, Relative entropy]
Let $n\in\MZ_{\geq 1}$ be a positive integer and let $p,q\in\Delta^n$ with $q_i\neq 0$ for all $p_i\neq 0$. Then 
\begin{enumerate}
\item the (Shannon) {\em entropy} $\op{H}:\Delta^n\rightarrow\MR$ is the expected self-information $p\mapsto \sum_{i=1}^n p_i\log \frac{1}{p_i}$.
\item the {\em relative entropy} $D(p || q)$ is the function $\sum_{i\in \{1,2,\dots,n\}} p_i \log\frac{p_i}{q_i}$.
\end{enumerate}
\end{definition}

Relative entropy, also called the Kullback-Leibler divergence or information divergence, is a very familiar quantity in information theory and statistics. It has several well-known operational interpretations. We now present one of these operational interpretations adapted to the context of ensembles. 

\begin{remark}
Consider an ensemble of physical systems with a distribution $\pi$. Within this ensemble is a sub-ensemble of systems with a distribution $p$. For example, the distribution $p$ might represent the distribution $\pi$ conditioned on an observation taking a particular value. A sample system $S$ from the sub-ensemble is given to two observers $X$ and $Y$. Observer $X$ only knows that the system $S$ was sampled from the distribution $\pi$. Observer $Y$ knows that the system $S$ was sampled from the sub-ensemble $p$. For example, $Y$ would know the value of an observation on the system, whereas $X$ would not.

We claim that observer $Y$ possesses precisely $D(p || \pi)$ more nats of information about system $S$ than observer $X$. (This statement may be interpreted in expectation over many runs, if one wishes to avoid a Bayesian interpretation.)

To see this, note that the surprise to observer $X$ upon learning the precise state of the system is $\sum p_i \log \frac{1}{\pi_i}$ nats in expectation. The surprise to observer $Y$ upon learning the precise state of the system is $\sum p_i \log\frac{1}{p_i} = H(p)$ nats in expectation. Thus the extra information that observer $Y$ has over observer $X$ from knowing that the system $S$ is from the subensemble $p$ must equal $\sum p_i\log\frac{1}{\pi_i} - H(p) = \sum p_i \log \frac{p_i}{\pi_i} = D(p || \pi)$.

In particular, if $\pi$ is an equilibrium distribution (by this, we mean a steady-state distribution of some process that satisfies detailed balance), and system $S$ is drawn from an equilibrium ensemble, and $p$ is the new distribution obtained after conditioning on the result of a measurement on system $S$, then the amount of information obtained by that measurement is $D(p || \pi)$.
\end{remark}

\begin{remark}
Following Jaynes, we may define equilibrium by the condition that systems at equilibrium are maximally uninformative. Hence, we may adopt the convention that we have zero information about systems at equilibrium. With this convention, one may say that for a system with equilibrium distribution $\pi$, drawn from a $p$-ensemble, the relative entropy $D(p || \pi)$ represents the amount of information one knows about the system. We elevate this convention to a postulate.

\textbf{Postulate}: The information an observer knows about a system with equilibrium distribution $\pi$ drawn from a $p$-ensemble is $D(p||\pi)$ nats.\

This postulate may be considered a mathematical definition of the information an observer knows about a system. Note some consequences of this postulate:
	\begin{itemize}
	\item Since $D(p||\pi)$ is always non-negative, the amount of information an observer knows about the system can never be negative, and is zero iff $p = \pi$.
	
	
	\item If the system has two states, and each state is equally likely at equilibrium, and $p$ is the distribution $(1,0)$ then an observer has $1 \log \frac{1}{1/2} + 0 \log\frac{0}{1/2}=\log 2$ nats (or $1$ bit) of information.

	\item If the system has two states, and the equilibrium distribution is $(\pi_1,\pi_2)$, and $p$ is the distribution $(1,0)$ then an observer has $\log \frac{1}{\pi_1}$ nats of information. This is consistent with Shannon's formula for self-information. Intuitively, being in the possession of information that would greatly surprise an equilibrium observer equates to being well-informed.
	
	
	\item {\em The reader should feel free to skip this bullet on a first read.} Consider two correlated systems, call them $S_1$ and $S_2$, with joint distribution $p$, marginals $p_1$ and $p_2$, and equilibrium distributions $\pi_1$ and $\pi_2$. Define $\pi = \pi_1 \otimes\pi_2$, the tensor product distribution. By our postulate, the information in the joint system is given by $D(p || \pi)$. The information in $S_1$ and $S_2$ correspond to $D(p_1 || \pi_1)$ and $D(p_2 || \pi_2)$ respectively. There is also information in the correlation between the two systems, captured by the mutual information $I(S_1,S_2):=D(p || p_1 \otimes p_2)$. It is a well-known identity in information theory that $D(p || \pi) = D(p_1 || \pi_1) + D(p_2 || \pi_2) + I(S_1,S_2)$. This corresponds to our intuitive idea that the total information equals the information in $S_1$ plus the information in $S_2$ plus the information in the correlation.
	
	\end{itemize}

\end{remark}

\begin{remark}
Relative entropy measures the information that the observer knows about the system. Shannon entropy measures the information ``in'' the system, i.e., how much one would learn if one was told the exact state of the system. Jaynes~\cite{jaynes1957information} also refers to Shannon entropy as missing information. So a more random system has more Shannon entropy, and in this sense Shannon entropy is a measure of randomness, whereas relative entropy is a measure of information.
\end{remark}

\subsection{Bits and bit operations.}\label{ss:bitops}
For now, we will describe a bit by a Bernoulli random variable. (We will describe bits in more detail in a follow-up paper, where the origin of this Bernoulli random variable will become clear.) We will further assume that equilibrium corresponds to the distribution $\pi = (1/2,1/2)$. Note that $D(p || \pi) - D(q || \pi) = H(q) - H(p)$. So the increase in entropy relative to $\pi$ equals the decrease in Shannon entropy.

We will be particularly interested in three types of Bernoulli random variables, that we call $\mathbf{0},\mathbf{1}$, and $\textbf{*}$. The random variable of type $\mathbf{0}$ (respectively $\mathbf{1}$) takes the value $0$ (respectively $1$) with probability $1$. The random variable of type $\textbf{*}$ is equally likely to take the values $0$ and $1$. If two random variables $\textbf{*}_1$ and $\textbf{*}_2$ are correlated, by which we mean $\op{Pr}[\textbf{*}_1 = \textbf{*}_2] = 1$, we represent this by a dash $\textbf{*}_1 --- \textbf{*}_2$. If they are anticorrelated, by which we mean $\op{Pr}[\textbf{*}_1 = \textbf{*}_2] = 0$, we represent this by a cross $\textbf{*}_1 -\times - \textbf{*}_2$.

We define the following operations.

\begin{enumerate}
\item \textbf{Erasing} is defined as any operation that takes a random variable of type \textbf{*} to the random variable $\mathbf{0}$. Thus erasing is the process of replacing a random bit by a bit in the state $0$. In the process of erasing, we gain one bit of information about the system. This can be counterintuitive since one expects erasing to destroy information. A helpful way to think about this is that what is erased is not information, but {\em randomness} in the system. For example, when we erase a blackboard with chalk marks on it, we are taking the system from many possible states to a single state of no chalk marks, or all $0$'s. 

\item \textbf{Copying in the sense of Szilard} is defined as any operation that takes the pair $(\textbf{*}_1,\textbf{*}_2)$ of independent random bits to the correlated pair $(\textbf{*}_1 --- \textbf{*}_3)$. In other words, the first bit is left unaltered, while the second bit is evolved to contain a copy of the first bit.

\item\textbf{Copying in the sense of Landauer} is defined as any operation that takes the pair $(\textbf{*}_1,\mathbf{0})$ to the correlated pair $(\textbf{*}_1 --- \textbf{*}_3)$.

\item\textbf{NOT} is defined as any operation that takes the random variable $\textbf{*}_1$ to a new random variable $\textbf{*}_2$ where $\textbf{*}_1 -\times - \textbf{*}_2$. 

\item\textbf{Switching($0\to 1$)} is defined as any operation that takes the random variable $\textbf{0}$ to the random variable $\textbf{1}$. Switching($1\to 0$) is defined likewise.

\item \textbf{Randomizing} a bit is defined as any operation that takes a random variable $\textbf{0}$ to a random variable of type \textbf{*}. 
\end{enumerate}

Both erasing, as well as copying in the sense of Szilard, reduce Shannon entropy (and hence increase entropy relative to $\pi$) by $1$ bit, or $\log 2$ nats. Note that if we had only asked for these operations to succeed with some probability, Shannon entropy would reduce by less than $1$ bit. Copying in the sense of Landauer, NOT, Switching($0\to 1$) and Switching($1\to 0$) do not change Shannon entropy. Randomizing a bit increases Shannon entropy (and decreases relative entropy) by $1$ bit, or $\log 2$ nats. 

\section{The Szilard-Landauer Correspondence: A Special Case.}
The Szilard-Landauer correspondence (Theorem~\ref{lem:SL}) implies that operations that increase relative entropy require energy, and operations that reduce relative entropy can be exploited to do useful work. Before getting to the general case, we illustrate the ideas behind the Szilard-Landauer correspondence with the simple example of the Szilard engine, which allows easy calculations.

\subsection{The Szilard engine}
A ``Szilard engine'' consists of a single molecule of an ideal gas in a cylindrical vessel~\cite{Szilard1929}. We first calculate the work required to compress one molecule of an ideal gas from volume $V$ to volume $V/2$. (We will ignore stochastic effects. Since this example is intended more as an aid-to-intuition, this lack of rigour is pardonable.) In the isothermal limit, using the ideal gas law $PV = k_B T$, the work done on the gas equals
\[
W = - \int_{V}^{V/2} P dV = \int_{V/2}^V kT \frac{dV}{V} = kT \log V|_{V/2}^V = kT\log 2.
\]

Following Szilard, we will think of the position of the molecule as representing a bit. So we label the bit as being in the state ``0'' precisely when the molecule is in the left half of the vessel. When the molecule is in the right half of the vessel, we label the bit as being in the state ``1.'' We will represent this bit by a Bernoulli random variable $X$, so that $Pr[X=0] = 1 - Pr[X=1]$. Note that at equilibrium, $Pr[X=0] = Pr[X=1] = 1/2$.

\paragraph{Randomizing a bit:} We illustrate, following Landauer, how useful work can be obtained by randomizing a bit. Suppose that initially the bit is of type $\mathbf{0}$. We insert a partition at the center and attach a weight via a pulley to the partition. The molecule collides against the partition due to random motion, causing the weight to be lifted, thus doing useful work. To maximize the amount of useful work that can be obtained, one works in the isothermal limit. This means that The weight is carefully chosen to be infinitesimally lighter than the expected force being exerted by the molecule on the partition. In this case, the work done by the gas is {\em no more than} $kT\log 2$. In thermodynamics, the work that can be obtained by a system is quantified by free energy, hence we have observed a correspondence between information and free energy.

At the end of the lifting process, the single molecule of ideal gas could be anywhere in the volume $V$ of the cylindrical vessel, so the bit has changed its type from $\mathbf{0}$ to $\textbf{*}$. In other words, the bit has been randomized. Importantly, note that if the bit had been of type $\mathbf{1}$ then the weight would have been lowered, and the work done would have been negative, {\em unless} we had reconnected the pulley and the weight to the other side. So the information about the molecule's position is crucial in deciding the side to which the pulley and the weight connect. In particular, if one has no information about which side the molecule is on, then in expectation no work can be obtained from the system! 

By a similar calculation, the work required to ``erase'' the bit, so that a molecule that was initially uniformly distributed through the entire volume is confined to the left half of the vessel, equals the work required to compress one molecule of an ideal gas to half its volume, which equals {\em no less than} $k_B T\log 2$. 

\paragraph{Maxwell's demon:}
Though it was not our purpose, we have arrived at a position whence we can succinctly state Szilard's resolution to Maxwell's demon, as well as Landauer's clarification. If one wants to measure which side of the vessel the molecule is on, one needs to copy the bit $X$ to another bit $Y$ (implicitly assumed by Szilard to be of type $\textbf{*}$) which is part of the measuring apparatus. After performing the measurement, we are now in possession of information about the bit $X$, in the form of correlation (mutual information) between $X$ and $Y$. This information can be randomized to extract $k_B T \log 2$ units of energy in the reversible, isothermal limit. Therefore, Szilard concluded that preservation of the laws of thermodynamics requires that the act of copying must cost at least $k_B T\log 2$. (Of course, this holds only for perfect copying. For copying with some errors, a weaker bound is obtained.)

The power of this argument of Szilard's is its abstractness: it is blind to the actual dynamical implementation of Maxwell's demon: whether it be a trapdoor, or a camera, or a biological organism. So long as extracting useful work requires knowledge of the position of the molecule, Szilard's argument will hold. This can be contrasted with Smoluchowski's insightful partial resolution of Maxwell's demon, by calculations with the specific example of a ratchet and pawl mechanism~\cite{smoluchowski1912experimental}.

One complaint against Szilard's resolution is that it ``merely'' pushes the problem back one step. After all, how does one know that knowledge of the position of the molecule is required to do useful work, except from the second law? So isn't the argument circular? 

In our opinion, this complaint is correct, but also naive. The goal here is not to prove the second law. The advance made by Szilard consists in replacing the statement ``in expectation, Maxwell's demon can not reduce entropy'' by the more falsifiable statement ``in expectation, no useful work can be obtained from Szilard's engine without information about the position of the molecule.'' 

Landauer's work clarified that if copying is performed in the sense of Landauer, and not in the sense of Szilard --- so that $Y$ is initialized to a state of zero entropy --- then the laws of thermodynamics are safe even if the act of copying requires no energy. This is because, after copying, both $X$ and $Y$ are random, but correlated. Now we randomize $X$ to obtain energy, and we go to a state where $X$ and $Y$ are both random, and uncorrelated. Thus, in this case, the work done is being obtained by randomizing $Y$, and the laws of thermodynamics are not threatened.

Actually neither Szilard nor Landauer precisely state what they mean by copying. Their notions of copying are implicit, and treated as self-evident. When Landauer states that there is no cost to copying, he is perhaps too hasty to attribute an error to Szilard, without recognizing the possibility that Szilard may have conceptualized copying in a sense different from Landauer's.

Whether (or how closely) these bounds of $k_B T\log 2$ and $0$ for copying in the senses of Szilard and Landauer are achievable is a different, altogether more subtle question, especially once one abandons reversible and isothermal limits. We will get to this question only in our next paper.

\section{The General Case}
We now state and prove a much more general correspondence between information and free energy. We have already identified the relative entropy $D(p || \pi)$ with the information we possess about a system with equilibrium distribution $\pi$. The Szilard-Landauer correspondence allows us to generalize the above calculation to talk about erasing, and randomizing, information in this abstract sense. It does so by relating available free energy to $k_B T$ times the relative entropy (Theorem~\ref{lem:SL}).

\begin{definition}
Let $n\in\MZ_{\geq 1}$. Given ``energies'' $E:\{1,2,\dots,\}\rightarrow\MR$ and a ``temperature'' $T\in\MR_{\geq 0}$, define
\begin{enumerate}
\item the {\em average energy} $\langle E \rangle_p := \sum_{i=1}^n p_i E(i)$, where $p\in\Delta^n$.

\item the {\em free energy} $F_{E,T}:\Delta^n\rightarrow \MR$ by $p \mapsto \langle E \rangle_p - k_BT \op{H}(p).$

\item the {\em partition function} $Z_{E,T}:= \sum_{i=1}^n e^{-\frac{E(i)}{k_BT}}.$

\item the {\em Gibbs distribution} $\pi_{E,T}:=\frac{1}{Z_{E,T}}(e^{-\frac{E(1)}{k_BT}},e^{-\frac{E(2)}{k_BT}},\dots,e^{-\frac{E(n)}{k_BT}}).$

\end{enumerate}
\end{definition}

\begin{remark}
We have chosen to restrict to finite spaces merely for matters of pedagogy, to bring out the key ideas without distracting technical details. Everything that follows should go through in a more general setting where the finite state space is replaced by a manifold, Shannon entropy by differential entropy, and summation by integration.
\end{remark}

\begin{remark}
The thermodynamic notion of free energy is that it equals the maximum work that can be extracted from the system.  A system has ``internal energy,'' but it also has entropy, or randomness. Before one can access the energy, one has to pay to reduce the randomness. In other words, the maximum useful work one can extract equals the average energy minus the average randomness. But randomness is dimensionless, so one has to multiply it by $k_B T$.
\end{remark}

\begin{remark}
Our free energy is a restriction to finite spaces of the free energy functional in \cite[Equation~(5)]{jordan1998variational}. As noted there, this form of the free energy functional is a Lyapunov function for the Fokker-Planck equation, i.e., it satisfies an analog of Boltzmann's H-theorem. This is another sense in which we find it justified to call this function as ``free energy.''
\end{remark}

\begin{theorem}[Szilard-Landauer Correspondence]\label{lem:SL}
Let $n\in\MZ_{\geq 1}$. Let $E:\{1,2,\dots,n\}\rightarrow\MR$ and $T\in\MR_{\geq 0}$. Let $\pi=\pi_{E,T}$ be the Gibbs distribution. Then $F_{E,T}(p) - F_{E,T}(\pi) = k_BT D(p || \pi)$ for all $p\in\Delta^n$.
\end{theorem}
\begin{proof}
Let $F=F_{E,T}$. Then
\begin{align*}
   F(p) &= \langle E \rangle_p - k_BT \op{H}(p)\
\\ &= \sum_{i=1}^n p_i E_i + k_B T p_i \log p_i
\end{align*}

Fix $i\in\{1,2,\dots,\}$. We have 
\begin{align*}
\log \pi_i &= -\log Z - \frac{E_i}{k_BT}\
\\\Rightarrow E_i &= - k_B T\log\pi - k_B T\log Z\
\\\Rightarrow p_i E_i + k_B T p_i \log p_i &=  k_B T( - p_i\log\pi - p_i\log Z + p_i\log p_i)\   
\end{align*}

Summing the last expression over all $i$, we have:

\begin{align*}
F(p) &= k_BT D(p || \pi) - k_B T \log Z
\end{align*}

It is now enough to show that $F(\pi) = -k_B T \log Z$. This is true because
\begin{align*}
F(\pi) &= \langle E\rangle_\pi - k_BT H(\pi)\
\\ &= \sum_{i=1}^n E_i\pi_i + k_BT \pi_i \log \pi_i\
\\ &= \sum_{i=1}^n E_i\pi_i + k_BT \pi_i \log e^{-\frac{E_i}{k_BT}} - \pi_ik_BT \log Z\
\\ &= \sum_{i=1}^n E_i\pi_i - k_BT\pi_i\frac{E_i}{k_BT} - \pi_ik_BT\log Z\
\\ &= - k_B T\log Z.
\end{align*}
\end{proof}

\begin{remark}\label{rmk:landprin}
The identification between information processing and thermodynamic tasks obtained from the Szilard-Landauer principle allows a more abstract style of analysis of information-processing systems. It allows us to make statements of the form ``if an operation is erasing $n$ bits of information then it must require at least $n k_B T\log 2$ units of energy,'' without requiring a detailed understanding of the dynamics of the system doing the information processing. From the calculations in Subsection~\ref{ss:bitops}, we can now state that erasing, as well as copying in the sense of Szilard, must cost at least $k_B T\log 2$ energy. No lower bound can be obtained for copying in the sense of Landauer, NOT, and Switching. By randomizing a bit, we can obtain no more than $k_B T\log 2$ energy. Thus we recover as special cases the bounds on the cost of information processing from the works of Szilard, Landauer, and Bennett, based on analyses of Maxwell's demon.
\end{remark}

\begin{remark}\label{rmk:2ndlaw}
With this equivalence between relative entropy and free energy in hand, one can ask what processes reduce relative entropy (i.e., satisfy the Second Law). Answers to this, and related questions, can be found in \cite{Cover1994SecondLaw} and \cite[2.9]{cover2012elements}. In particular, for Markov chains, {\em even if detailed balance does not hold}, so long as there exists a stationary distribution, entropy relative to the stationary distribution is non-increasing. This yields a very powerful hammer against proposals for Maxwell's demons: ask if the dynamics of the joint system can be described in a Markovian manner, and if a stationary distribution exists. If these conditions are satisfied, then the proposal won't be able to violate the Second Law.
\end{remark}

\begin{remark}
Theorem~\ref{lem:SL} provides an equivalence between information as measured by relative entropy, and the available free energy $F(p) - F(\pi)$. To express this more vividly, let us use the term ``battery'' to informally denote a system that is a store of free energy. Then the theorem says that every battery is a system about which we know information, and every known bit of information can be viewed as available free energy. Charging a battery corresponds to erasing a bit, and requires energy to be supplied to the system. Discharging a battery corresponds to randomizing a bit, and useful work can be obtained in this process. The below table summarizes this correspondence.
\[
\begin{array}{ccc}
\text{Information} &\leftrightarrow &\text{Energy}\
\\\text{Relative Entropy} &\leftrightarrow &\text{Available Free Energy}\
\\\text{Bits} &\leftrightarrow &\text{Batteries}\
\\\text{Erasing a bit}&\leftrightarrow&\text{Charging a battery}\
\\\text{Randomizing a bit}&\leftrightarrow&\text{Discharging a battery}\
\\\text{Unreliable bit}&\leftrightarrow&\text{Leaky battery}
\end{array}
\]
\end{remark}

\begin{remark}
Following Bennett~\cite{bennett82thermodynamics}, statements referring to the observation that erasing a bit costs at least $k_B T\log 2$ are commonly referred to as Landauer's principle. We are proposing that a more appropriate name would be the Szilard-Landauer correspondence, to acknowledge Szilard's pioneering work in this direction, as well as to recognize that these ideas constitute more than a principle. In fact, they provide a dictionary between information processing and thermodynamics as we have remarked.
\end{remark}

\begin{remark}
A somewhat radical interpretation of the Szilard-Landauer correspondence is that it demystifies energy, by revealing it to be information. Information often has a combinatorial interpretation, whereas to our eyes energy is a somewhat more arcane quantity. Taking this point of view to its logical conclusion would require reinterpreting every occurrence of energy in physics in terms of information. In some sense, the combinatorial point of view is not so new and radical, and has been available since Maxwell pioneered the view of thermodynamics as a statistical consequence of classical mechanics. It has certainly been mathematically exploited on numerous occasions, as evidenced by the central role played by partition functions. However, the mental picture of the worker in statistical mechanics continues to be in terms of energy. Augmenting this mental picture with one in terms of information may have some advantages in terms of pedagogy, and also when it comes to putting statistical mechanics on common ground with other areas of information systems theory like statistical inference and machine learning with which it shares many common techniques.
\end{remark}

\begin{remark}
It has been pointed out by several authors~\cite{lecerf1963machines,bennett73logical, fredkin1982conservative} that loss of information is not essential to the logical operations employed in the computing process. Therefore, the Szilard-Landauer correspondence provides no non-trivial lower bound to the cost of computing. It is similarly ineffective in providing lower bounds to the cost of switching. If one actually believes that computation (and hence switching) can be done for zero or very low cost, then this needs to be shown by a constructive argument that is realistic enough to be approximately implementable. On the other hand, if one believes that computation (and hence switching) need substantial amounts of energy, then one still needs to investigate richer models, in an effort to make transparent precisely where the cost is incurred. Either belief requires us to investigate less ideal, more detailed models, which is what we start doing in our next paper.
\end{remark}

\paragraph{Acknowledgments:} I thank Pulkit Grover from Carnegie Mellon University,  Nick Jones from Imperial College London, and Rahul Dandekar from TIFR Mumbai for helpful discussions.  
\bibliographystyle{amsplain}
\bibliography{../../eventsystems}
\end{document}